\documentclass{llncs}
\usepackage{llncsdoc}
\usepackage{stmaryrd,color}
\begin{document}
\title{Monoid automata for displacement context-free languages}
\author{Alexey Sorokin \inst{1,2}}
\institute{Moscow State University, Faculty of Mathematics and Mechanics
\and Moscow Institute of Physics and Technology, \\Faculty of Innovations and High Technologies}
\maketitle
\begin{abstract}
In 2007 Kambites presented an algebraic interpretation of Chomsky-Sch\"{u}tzenberger theorem for context-free languages. We give an interpretation of the corresponding theorem for the class of displacement context-free languages which are equivalent to well-nested multiple context-free languages. We also obtain a characterization of $k$-displacement context-free languages in terms of monoid automata and show how such automata can be simulated on two stacks. We introduce the simultaneous two-stack automata and compare different variants of its definition. All the definitions considered are shown to be equivalent basing on the geometric interpretation of memory operations of these automata.
\end{abstract}

\newcommand{\inbr}[1]{\{ #1 \}}
\newcommand{\kav}[1]{``#1''}
\newcommand{\inang}[1]{\langle #1 \rangle}
\newcommand{\hide}[1]{}

\newcommand{\aar}{\mathrm{ar}}
\newcommand{\Bb}{\mathcal{B}}
\newcommand{\Aa}{\mathcal{A}}

\section{Introduction}
In last decades the theory of monoid automata attracts a great interest both from the specialists in the theory of formal languages and algebra. The first are looking for a fine algebraic characterization of formal languages, which simplifies studying their properties and shows known theoretical facts in a wider scope. The algebraists are interested in the questions of effective computations in groups and semigroups where different variants of automata can be useful. Also the theory of monoid automata has some connections with the combinatorial group theory, e.g. in studying word problems for groups. For a more detailed survey and references see \cite{Kambites} or \cite{Zetzsche2011}.

A monoid automaton (or valence automaton) is a finite automaton augmented with a register storing an element of a particular monoid. Each transition of the automaton multiplies the current element in the memory by the monoid element associated with this transition. The automaton accepts a word if it reaches a final state after reading the word with the monoid unit in the register. The usage of the memory register allows to recognize more complex languages then the usual automata do. Evidently, the recognizing power essentially depends from the monoid serving as the register.

There is a straightforward approach to present a characterization of a language family in case there exists some another formal model of memory and computation for processing it. Assume that there is no content in the memory before reading the word and the memory storage should also be empty in the end in case the word is recognized. This condition holds for standard models of computation such as pushdown automata or embedded pushdown automata as well as for many other models. Consider the composition of all the operations executed during the successful computation, it obviously equals the identity element. So in this case it suffices to consider the monoid of admissible memory operations to provide a monoid automaton for a given family of languages.

For example, for the family of context-free languages every admissible operation is a composition of pushing and popping some symbols from the stack. The monoid of such operations is just the polycyclic monoid $\Bb_n = \inbr{x_i, \overline{x}_i \mid 1 \leq i \leq n}^* / \{ x_i\overline{x}_i = 1 \mid 1 \leq i \leq n\}$ (\cite{NivatPerrot}) where $x_1, \ldots, x_n$ are the elements of the stack alphabet and the equality $x_i\overline{x}_i = 1$ reflects the fact that popping $x_i$ immediately after pushing it on the stack is the same as doing nothing. But if the computation model is more complicated such as embedded pushdown automata for tree-adjoining languages this approach is useless because we cannot recover the structure of the monoid of operations.

The alternative approach to solve this problem is given in the work of Kambites (\cite{Kambites}). He shows that in the case of context-free languages we may use the Chomsky-Sch\"{u}tzenberger theorem, which states that every context-free language is the rational transduction of the language of correct bracket sequences. By Kambites theorem it suffices to find a monoid with an identity language isomorphic to the set of correct bracket sequences and use its elements as memory contents. It is not very interesting in the case of context-free languages because such a monoid is unsurprisingly a polycyclic monoid but very useful in general since Chomsky-Sch\"{u}tzenberger theorem is known for different families of languages.

In our work we consider the family of displacement context-free languages (\cite{Sorokin}) which coincides with the family of well-nested multiple context-free languages. Some computer scientists offer them as a possible candidate to formalize the notion of mildly context-sensitive language (see \cite{Kanazawa2011}). Chomsky-Sch\"{u}tzenberger theorem for this family of languages is proved in \cite{Yoshinaka}. We use this theorem to give a characterization of displacement context-free languages in terms of monoid automata and then show how the element of the constructed monoid are interpreted as operations on two stacks.

The paper is organized as follows: first we recall the definition of a monoid automaton and present the Kambites theorem. Then we define the family of displacement context-free languages and formulate the Chomsky-Sch\"{u}t\-zen\-ber\-ger theorem for it. We interpret the multibracket sequences from this theorem as the identity language of some monoid which gives us the characterization of displacement context-free languages in terms of monoid automata. Afterwards we show that this monoid is isomorphic to a particular submonoid of the cartesian product of two polycyclic monoids which allows us to interpret its elements as operations on the pair of stacks. Then we study some variants of the obtained computation model and show that its recognizing power does not depend on the possibility to observe the contents of the stacks before executing the command and other minor modifications of its definition.

\section{Preliminaries}

\subsection{Monoid automata}
In this section we introduce the definitions and concepts which would be useful in the further. We expect the reader to be familiar with basic notions of formal languages theory, such as finite automata, rational transductions and context-free grammars, also some knowledge of semigroup theory is required. For necessary information refer to any textbook on formal languages theory, such as \cite{Salomaa}, also see \cite{Berstel} and \cite{Lallement} for the introduction into the theory of rational transductions and theory of semigroups respectively. In this section we focus the attention on monoid automata and its interconnections with other objects of formal languages theory.

\begin{definition}
A monoid automaton ($M$-automaton) over the alphabet $\Sigma$ is a tuple $\mathcal{A} = \inang{ Q, \Sigma, M, P, q_0, F }$ where $Q$ is a finite set of states, $\Sigma$ is a finite alphabet, $M$ is a partial monoid with the identity $1$, $q_0 \in Q$ is an initial state, $F \subseteq Q$ is the set of final states and $P \in Q \times M \times (\Sigma \cup \epsilon) \times Q$ is a set of transitions.
\end{definition}

Just in the case of finite automata the notion of a label can be extended from edges to paths in the automaton. The only difference is that we replace mere concatenation by the multiplication operation of the monoid. According to this definition, the usual finite automata are $\mathrm{1}$-automata. Note that in all the cases we consider nondeterministic automata, which means we allow multiple moves with the same label in one state.
Also note that monoid automata are blind in the sense that they do not take the current element in the memory into account before multiplying it by the element associated with an edge.

\begin{definition}
A word $w \in \Sigma^*$ is accepted by the $M$-automaton $\!\mathcal{A}\!=\!\inang{\!Q, \Sigma, M, P, \linebreak q_0, F }$ iff there is a state $q \in Q$ such that the pair $\inang{ 1, w }$ labels some path from $q_0$ to $q$. The language recognized by the automaton $\mathcal{A}$ is denoted by $L(\mathcal{A})$.
\end{definition}

\begin{example}
 Let $S_1$ be a monoid with the generators $\{\alpha, \overline{\alpha}\}$ and the defining relation $\alpha \circ_1 \overline{\alpha} = 1$ and $S_2$ be a monoid with the generators $\{\beta, \overline{\beta}\}$ and the defining relation $\beta \circ_2 \overline{\beta} = 1$. Then the language $\{a^n b^n c^n \mid n \in \bbbn^+\}$ is recognized by an automaton $\mathcal{\alpha} = \inang{ \{q_0, q_1, q_2\}, \{a, b, c\}, S_1 \times S_2, P, q_0, \{q_2\} }$ where \linebreak $P = \{\inang{ q_0, \inang{ \alpha, 1 }, a, q_0}, \: \inang{ q_0, \inang{ \overline{\alpha}, \beta }, b, q_1}, \: \inang{ q_1, \inang{ \overline{\alpha}, \beta }, b, q_1}, \: \inang{ q_1, \inang{ 1, \overline{\beta} }, c, q_2}, \linebreak\inang{ q_2, \inang{ 1, \overline{\beta} }, c, q_2}\}$.
\end{example}

Let $M$ be a finitely generated monoid and $X$ be its system of generators. The identity language of $M$ consists of all the words in $X^*$ that represent identity. The proposition below enlightens the connection between monoid automata and finite transducers. It entails that the class of languages recognized by $M$-automata is closed under rational transductions for any finitely generated monoid $M$.

\begin{proposition}[Kambites, 2007]
The following conditions are equivalent:
\begin{enumerate}
\item $L$ is accepted by an $M$-automaton.
\item $L$ is a rational transduction of the identity language of $M$ with respect to some finite generating set $X$.
\item $L$ is a rational transduction of the identity language of $M$ with respect to every finite generating set $X$.
\end{enumerate}
\end{proposition}

This result of Kambites offers a powerful method of characterizing a particular family of languages in terms of monoid automata, if this family is closed under rational transductions. To prove that all languages recognized by $M$-automata are, for an instance, context-free, it suffices to construct a context-free grammar for the identity language of $M$. To prove the opposite inclusion one may either present an automata characterization of the language family (pushdown automata give such characterization for context-free languages) and then translate it to the language of monoid automata or find some ``typical'' language in the family, such that all other languages are its images under rational transductions, and show that this language recognized by an $M$-automaton for the monoid $M$ under consideration.

For context-free languages it is reasonable to use Chomsky-Sch\"{u}tzenberger theorem. The Dyck language of rank $n$ is a language containing all correct bracket sequences on $n$ types of brackets $a_1, \overline{a}_1, \ldots, a_n, \overline{a}_n$. It is generated by a context-free grammar with the rules $S \to a_i S \overline{a}_i S, \: S \to \epsilon$, where $i$ ranges from $1$ to $n$. The next theorem shows that it is in some sense ``typical'' among the context-free languages:

\begin{theorem}[Chomsky-Sch\"{u}tzenberger, \cite{Chomsky}]
A language $L$ is context-free if and only if it is a rational transduction of the language $D_n$ for some $n \in \bbbn$.
\end{theorem}

The proof of this theorem can be found, for example, in \cite{Lallement}. Informally, the statement of the theorem roughly corresponds to the fact that the subtrees of the derivation tree are either embedded one into another or do not intersect. Now we want to show that in fact this theorem is about monoid automata.

Let $X$ be a finite set of generators, then for every element $x \in X$ we define two operators $P_x$ and $Q_x$ on the free monoid $X^*$. $P_x$ transforms a string $w$ to the string $wx$ simulating the push operation. $Q_x$ conversely transforms a string of the form $wx$ to the string $w$ and is a right inverse of $P_x$. The set of all $P_x, Q_x$ is extended to the submonoid $\mathcal{P}_X$ of the monoid of partial functions from $X^*$ to $X^*$. This monoid was first studied in the work of \cite{NivatPerrot} and plays a great role in the structural theory of semigroups.

Polycyclic monoid automata obviously are capable to perform the operations ``push'' and ``pop'' of usual pushdown automata, which suffices to simulate its work. Note that if we refer to the elements of $X$ as types of brackets, then $P_x$ naturally corresponds to the opening bracket, as well as $Q_x$ to the closing. With respect to this translation the identity language of the monoid $\mathcal{P}_X$ is exactly the set of correct bracket sequences. Summarizing, the following theorem holds:

\begin{theorem}[Kambites, 2007]\label{Kambites}
The following conditions are equivalent:
\begin{enumerate}
\item The language $L$ is context-free.
\item The language $L$ is recognized by some polycyclic monoid automata.
\end{enumerate}
\end{theorem}
Note that this theorem can also be proved directly without any references to Chomsky-Sch\"{u}tzenberger theorem, just in the same way like the equivalence between context-free grammars and pushdown automata is established.

\subsection{Displacement context-free grammars}
In this section we define the generalization of context-free grammars, the displacement context-free grammars (DCFGs), introduced in \cite{Sorokin}. They are just another realization of well-nested multiple context-free grammars (wMCFGs) but are more convenient for the purposes of our work. It is worth noting that wMCFGs are thoroughly studied in last years, for example in \cite{Kanazawa2009} or \cite{Kanazawa2011}.

Let $\Sigma$ be a finite alphabet and $1$ be a distinguished separator, $1 \notin \Sigma$. For every word $w \in (\Sigma \cup 1)^*$ we define its rank $rk(w) = |w|_1$. We define the $j$-th intercalation operation $+_j$ which consists in replacing the $j$-th separator in its first argument by its second argument. For example, $a1b11d+_2 c1c = a1bc1c1d$.

Let $k$ be a natural number and $N$ be the set of nonterminals. The function $rk \colon N \to \overline{0,k}$ assigns every element of $N$ its rank. Let $Op_k = \{ \cdot, +_1, \ldots, +_k\}$ be the set of binary operation symbols, then the ranked set of correct terms $Tm_k(N, \Sigma)$ is defined in the following way (we write simply $Tm_k$ when it causes no confusion):
\begin{enumerate}
    \item $N \subset Tm_k(N, \Sigma)$,
    \item $\Sigma^* \subset Tm_k(N, \Sigma), \: \forall w \in \Sigma^* \: rk(w) = 0$,
    \item $1 \in Tm_k, \: rk(1) = 1$,
    \item If $A, B \in Tm_k$ and $rk(A) + rk(B) \leq k$, then $(A \cdot B) \in Tm_k, \\ rk(A \cdot B) = rk(A) + rk(B)$.
    \item If $j \leq k, \: A, B \in Tm_k, \: rk(A) + rk(B) \leq k + 1, \: rk(A) \geq j$, then \\ $(A +_j B) \in Tm_k, \: rk(A \cdot B) = rk(A) + rk(B) - 1$.
\end{enumerate}

We will often omit the symbol of concatenation and assume that concatenation has greater priority then intercalation, so $Ab +_2 cD$ means $(A \cdot b) +_2 (c \cdot D)$. The set of correct terms includes all the terms of sort $k$ or less that also do not contain subterms of rank greater than $k$. The set of ground terms $GrTm_k(\Sigma)$ consists of all terms that have only elements of $\Sigma^* \cup \{1\}$ in its leafs. For every ground term we can calculate its value interpreting elements of $\Sigma \cup 1$ as themselves and the operation symbols $\cdot$ and $+_j$ as concatenation and $j$-th intercalation respectively. Let $\nu \colon GrTm_k(\Sigma) \to (\Sigma \cup 1)^*$ be the value function, then in this interpretation $rk(\alpha)$ equals $rk(\nu(\alpha))$ for every ground term $\alpha$. If we assign every nonterminal of rank $j$ an arbitrary word of rank $j$, the same interpretation holds for non-ground terms either.

\begin{definition}
A $k$-displacement context-free grammar ($k$-DCFG) is a quadruple $G = \inang{ N, \Sigma, P, S }$, where $\Sigma$ is a finite alphabet, $N$ is a finite ranked set of nonterminals and $\Sigma \cap N = \emptyset, S \in N$ is a start symbol such that $rk(S) = 0$ and $P$ is a set of rules of the form $A \to \alpha$. Here $A$ is a nonterminal, $\alpha$ is a term from $Tm_k(N, \Sigma)$, such that $rk(A) = rk(\alpha)$.
\end{definition}

A context $C[]$ is simply a term from $Tm_k$ with a distinguished placeholder $\#$ instead one of its leafs. If $\beta$ is a term, then $C[\beta]$ denotes the result of replacing $\#$ by $\beta$ (in case the resulting term is in $Tm_k$). For example, $C[] = b1 +_1 (a \cdot \#)$ is a context and $C[A \cdot c] = b1 +_1 aAc$.

\begin{definition}
The derivability relation $\vdash_G \in N \times Tm_k$ associated with the grammar $G$ is the smallest reflexive transitive relation such that the facts $(B \to \beta) \in P$ and $A \vdash C[B]$ imply that $A \vdash C[\beta]$ for any context $C$. Let the set of words derivable from $A \in N$ be $L_G(A) = \{ \nu(\alpha) \mid A \vdash_G \alpha, \: \alpha \in GrTm_k\}$. Then $L(G) = L_G(S)$.
\end{definition}

\begin{example}
Let the $i$-DCFG $G_i$ be the grammar $G_i = \inang{ \{S, T\}, \{a, b\}, P_i, S}$. Here $P_i$ is the following set of rules (the notation $A \to \alpha | \beta$ stands for \linebreak $A \to \alpha, A \to \beta)$:
$$\begin{array}{rcl}
    S & \to & \underbrace{(\ldots(}_{i - 1\mbox{ times}} aT +_1 a) + \ldots) +_1 a \; | \; \underbrace{(\ldots(}_{i - 1\mbox{ times}} bT +_1 b) + \ldots) +_1 b\\
    T & \to & \underbrace{(\ldots(}_{i - 1\mbox{ times}} aT +_1 1a) + \ldots) +_{i} 1a \; | \; \underbrace{(\ldots(}_{i - 1\mbox{ times}} bT +_1 1b) + \ldots) +_{i} 1b \; | \; 1^i
\end{array}$$
The grammar $G_i$ generates the language $\{ w^{i + 1} \mid w \in \{a, b\}^+\}$. For example, this is the derivation of the word $(aba)^3$ in $G_2$: $S \to (aT +_1 a) +_1 a \to (a((bT +_1 1b) +_2 1b) +_1 a) +_1 a \to (a((b((aT +_1 1a) +_2 1a) +_1 1b) +_2 1b)+_1 a) +_1 a \to (a((b((a11 +_1 1a) +_2 1a) +_1 1b) +_2 1b)+_1 a) +_1 a = (a(b(a1a1a +_1 1b) +_2 1b)+_1 a) +_1 a = (aba1ba1ba +_1 a) +_1 a = abaabaaba$.
\end{example}

We have already noted that $k$-DCFGs are equivalent to well-nested $(k+1)$-multiple context free grammars. In the case of $k = 1$ the intercalation operation is simply the wrapping operation of head grammars (\cite{Pollard1984}, \cite{Roach1987}), which are equivalent to tree adjoining grammars (TAGs), as proved in \cite{Vijay-SchankerWeirJoshi1986}. We will not recall the definitions of these classes due to the lack of space. The interested reader may consult \cite{Seki} and \cite{Kanazawa2009} for the definitions of wMCFGs and \cite{JoshiSchabes} for the definition of TAGs.

Comparing the definition of DCFG with the definition of wMCFG it is necessary to mention that wMCFGs does not impose any condition on the rank of subterms which are well-nested substructures of the righthand side of the rule in terms of wMCFGs. However, this restriction can be also removed in the case of DCFGs: it is possible to show that for every term $\alpha$ which do not contain leaves of sort greater then $k$ and is of sort $k$ itself an equivalent term $\beta \in Tm_k(N, \Sigma)$ can be constructed. Equivalence in this case means that both terms have the same value under arbitrary assignment of values to nonterminals. We omit the details of the proof. So the condition on subterm ranks is redundant in general but we leave it for the sake of consistence.

\subsection{Chomsky-Sch\"{u}tzenberger theorem and correct multibracket sequences}

To present the Chomsky-Sch\"{u}tzenberger theorem we should replace brackets with multibrackets. Let $X$ be a ranked alphabet with the arity function $\rho \colon X \to \overline{1, L}$, where $L$ is a positive integer called the rank of $X$. We define the set of multibrackets $B(X) = \{ x^j, \overline{x}^j \mid x \in X, \: j \in \overline{1, \rho(x)}\}$. Let $w[j]$ denotes $j$-th letter in a word $w \in B(X)^*$ (the numeration starts with zero) and $Pos(w) = \{0, 1,\ldots, |w| - 1\}$.

\begin{definition}
$w \in B(X)^*$ is called a correct multibracket sequence if the set $Pos(w)$ can be partitioned into disjoint sets $H_1, \ldots, H_m$ such that: \newline
1) Every $H_t$ contains an even number of elements. If $i_1 < j_1 < i_2 < \ldots < i_r < j_r$ are the elements of some set $H_t$,then there exists an element $x \in X$ such that $r = \rho(x)$ and for every $l \leq r$ it holds
that $H[i_l] = x^l, H[j_l] = \overline{x}^l$. \newline
2) If $H$ and $H'$ are two sets in the partition and $i_1 < j_1 < \ldots < i_r < j_r$ and $i'_1 < j'_1 < \ldots < i'_s < j'_s$ are their elements, then one of the following alternatives holds:
\begin{itemize}
    \item $j_r < i'_1$ or $j'_s < i_1$,
    \item There exists $l \in \overline{1, r - 1}$ such that $j_l < i'_1 < j'_1 < \ldots < i'_s < j'_s < i_{l+1}$ or there exists $l' \in \overline{1, s-1}$ such that $j'_{l'} < i_1 < j_1 < \ldots < i_r < j_r < i'_{l'+1}$.
    \item For every $l \in \overline{1, r}$ there exists $l' \in \overline{1, s}$ such that $i'_{l'} < i_l < j_l < j'_{l'}$ or for every $l' \in \overline{1, s}$ there exists $l \in \overline{1, r}$ such that $i_l < i'_{l'} < j'_{l'} < j_l$.
\end{itemize}
\end{definition}

The generalized Dyck language $D(X)$ over the alphabet $X$ is the language of all correct multibracket sequences $w \in B(X)^*$. Informally, let the set $H$ in the partition consist of the positions $i_1 < j_1 < \ldots < i_s < j_s$. Let the elements of $H$ define a closed curve on the plane as it is shown on the figure below ($s = 3$), we refer to the set of such curves as the induced curves of the partition:
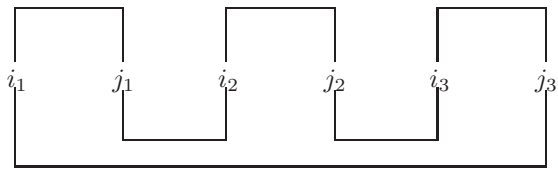
\begin{center}
\begin{picture}(300, 80)
\put(40, 40){$i_1$}\put(80, 40){$j_1$}\put(120, 40){$i_2$}\put(160, 40){$j_2$}\put(200, 40){$i_3$}\put(240, 40){$j_3$}
\put(43,50){\line(0,1){20}}\put(43,70){\line(1,0){41}}\put(84,70){\line(0,-1){20}}
\put(123,50){\line(0,1){20}}\put(123,70){\line(1,0){41}}\put(164,70){\line(0,-1){20}}
\put(203,50){\line(0,1){20}}\put(203,70){\line(1,0){41}}\put(244,70){\line(0,-1){20}}
\put(43,40){\line(0,-1){30}}\put(43,10){\line(1,0){201}}\put(244,10){\line(0,1){29}}
\put(84,39){\line(0,-1){19}}\put(84,20){\line(1,0){39}}\put(123,20){\line(0,1){20}}
\put(164,39){\line(0,-1){19}}\put(164,20){\line(1,0){39}}\put(203,20){\line(0,1){20}}
\end{picture}
\end{center}

Then $w$ is a correct multibracket sequence if it is possible to partite its set of positions in a way that the induced curves of this partition do not intersect. There is another geometrical intuition behind this definition: every set $H$ in the partition of correct multibracket sequence divides the sequence into its ``interior'' and ``exterior''. For any other set $H'$ in the partition there are four possibilities: $H'$ is in the interior of $H$; $H$ is in the interior of $H'$; $H'$ lies entirely in one of the intervals of the exterior of $H$; $H$ lies entirely in one of the intervals of the exterior of $H'$. Let $H$ consist of the elements $i_1 < j_1 < \ldots < i_s < j_s$ and $H'$ consist of $i'_1 < j'_1 < \ldots < i'_t < j'_t$. The picture below illustrates the possible variants of their mutual position ($s = 3$ and $t = 2$).

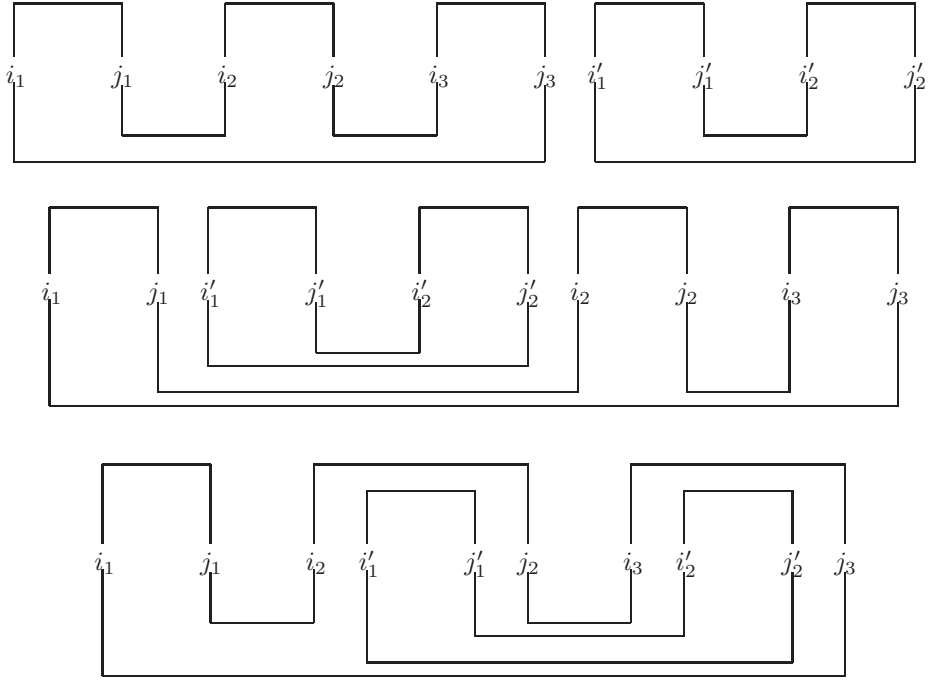
\begin{center}
\begin{picture}(347, 80)
\put(0, 40){$i_1$}\put(40, 40){$j_1$}\put(80, 40){$i_2$}\put(120, 40){$j_2$}\put(160, 40){$i_3$}\put(200, 40){$j_3$}
\put(3,50){\line(0,1){20}}\put(3,70){\line(1,0){41}}\put(44,70){\line(0,-1){20}}
\put(83,50){\line(0,1){20}}\put(83,70){\line(1,0){41}}\put(124,70){\line(0,-1){20}}
\put(163,50){\line(0,1){20}}\put(163,70){\line(1,0){41}}\put(204,70){\line(0,-1){20}}
\put(3,40){\line(0,-1){30}}\put(3,10){\line(1,0){201}}\put(204,10){\line(0,1){29}}
\put(44,39){\line(0,-1){19}}\put(44,20){\line(1,0){39}}\put(83,20){\line(0,1){20}}
\put(124,39){\line(0,-1){19}}\put(124,20){\line(1,0){39}}\put(163,20){\line(0,1){20}}

\put(220, 40){$i'_1$}\put(260, 40){$j'_1$}\put(300, 40){$i'_2$}\put(340, 40){$j'_2$}
\put(223, 50){\line(0,1){20}}\put(223,70){\line(1,0){41}}\put(264,70){\line(0,-1){20}}
\put(303, 50){\line(0,1){20}}\put(303,70){\line(1,0){41}}\put(344,70){\line(0,-1){20}}
\put(223,40){\line(0,-1){30}}\put(223,10){\line(1,0){121}}\put(344,10){\line(0,1){29}}
\put(264,39){\line(0,-1){19}}\put(264,20){\line(1,0){39}}\put(303,20){\line(0,1){20}}
\end{picture}

\begin{picture}(320, 80)
\put(0, 40){$i_1$}\put(40, 40){$j_1$}\put(200, 40){$i_2$}\put(240, 40){$j_2$}\put(280, 40){$i_3$}\put(320, 40){$j_3$}
\put(3,50){\line(0,1){25}}\put(3,75){\line(1,0){41}}\put(44,75){\line(0,-1){25}}
\put(203,50){\line(0,1){25}}\put(203,75){\line(1,0){41}}\put(244,75){\line(0,-1){25}}
\put(283,50){\line(0,1){25}}\put(283,75){\line(1,0){41}}\put(324,75){\line(0,-1){25}}
\put(3,40){\line(0,-1){40}}\put(3,0){\line(1,0){321}}\put(324,0){\line(0,1){39}}
\put(44,39){\line(0,-1){34}}\put(44,5){\line(1,0){159}}\put(203,5){\line(0,1){35}}
\put(244,39){\line(0,-1){34}}\put(244,5){\line(1,0){39}}\put(283,5){\line(0,1){35}}

\put(60, 40){$i'_1$}\put(100, 40){$j'_1$}\put(140, 40){$i'_2$}\put(180, 40){$j'_2$}
\put(63, 50){\line(0,1){25}}\put(63,75){\line(1,0){41}}\put(104,75){\line(0,-1){25}}
\put(143, 50){\line(0,1){25}}\put(143,75){\line(1,0){41}}\put(184,75){\line(0,-1){25}}
\put(63,40){\line(0,-1){25}}\put(63,15){\line(1,0){121}}\put(184,15){\line(0,1){24}}
\put(104,39){\line(0,-1){19}}\put(104,20){\line(1,0){39}}\put(143,20){\line(0,1){20}}
\end{picture}

\begin{picture}(300, 100)
\put(10, 40){$i_1$}\put(50, 40){$j_1$}\put(90, 40){$i_2$}\put(170, 40){$j_2$}\put(210, 40){$i_3$}\put(290, 40){$j_3$}
\put(13,50){\line(0,1){30}}\put(13,80){\line(1,0){41}}\put(54,80){\line(0,-1){30}}
\put(93,50){\line(0,1){30}}\put(93,80){\line(1,0){81}}\put(174,80){\line(0,-1){30}}
\put(213,50){\line(0,1){30}}\put(213,80){\line(1,0){81}}\put(294,80){\line(0,-1){30}}
\put(13,40){\line(0,-1){40}}\put(13,0){\line(1,0){281}}\put(294,0){\line(0,1){39}}
\put(54,39){\line(0,-1){19}}\put(54,20){\line(1,0){39}}\put(93,20){\line(0,1){20}}
\put(174,39){\line(0,-1){19}}\put(174,20){\line(1,0){39}}\put(213,20){\line(0,1){20}}

\put(110, 40){$i'_1$}\put(150, 40){$j'_1$}\put(230, 40){$i'_2$}\put(270, 40){$j'_2$}
\put(113, 50){\line(0,1){20}}\put(113,70){\line(1,0){41}}\put(154,70){\line(0,-1){20}}
\put(233, 50){\line(0,1){20}}\put(233,70){\line(1,0){41}}\put(274,70){\line(0,-1){20}}
\put(113,40){\line(0,-1){35}}\put(113,5){\line(1,0){161}}\put(274,5){\line(0,1){34}}
\put(154,39){\line(0,-1){24}}\put(154,15){\line(1,0){79}}\put(233,15){\line(0,1){24}}
\end{picture}

\end{center}

The next proposition offers $(rk(X) - 1)$-DCFG for $D(X)$, the proof follows from the definitions, so we omit it (we just reformulate the wMCFG-grammar from \cite{Yoshinaka} in terms of DCFGs):

\begin{proposition}
Let $X$ be the ranked alphabet of rank $L$, then the language of correct multibracket sequences over $X$ is generated the $(L-1)$-DCFG $G_{X} = \{ \{ S_i \mid i \in \overline{0,L-1}\}, B(X), P_{X}, S_0 \}$ where $P_X$ contains the following rules:
\begin{itemize}
    \item $S_{i + j} \to S_i S_j, \; i + j < L$,
    \item $S_{i + j - 1} \to S_i +_l S_j, \; i + j \leq L, \: l \leq i < L$,
    \item $S_{r}\!\to\!\!x^1\!\!\underbrace{(\ldots(}_{r\mbox{ times}}\!\!S_{r} +_1 (\overline{x}^1 1 x^2))\!+_2\! \ldots) +_{r} (\overline{x}^{r}1x^{r+1}))\overline{x}^{r+1}, \; x \in X, \: r = \rho(x) - 1$,
    \item $S_0 \to \epsilon, \; S_1 \to 1$.
\end{itemize}
\end{proposition}

Below we formulate the Chomsky-Sch\"{u}tzenberger theorem for the class of $k$-DCFGs. We omit the proof, since, as mentioned in \cite{Yoshinaka}, it can be recovered from the analogous theorem for the class of all MCFGs with natural modifications.

\begin{theorem}\label{Chomsky}
The language $L$ is a $k$-displacement context-free language if and only if it is a rational transduction of generalized Dyck language $D(X)$ for some alphabet $X$ of the rank $k + 1$.
\end{theorem}

\section{Monoid automata for displacement context-free grammars}\label{section-monoid-displ}
In this section we characterize the class of $k$-displacement context-free languages in terms of monoid automata.  For any set of $X$ of multibrackets we construct a monoid whose identity language is exactly $D(X)$ and then use Chomsky-Sch\"{u}tzenberger and Kambites theorems to prove the desired result.

Let $X$ be a generating set, $rk(X) = L$ and $\aar \colon X \to [1, L]$ be the arity function. Let $A$ be the set $A = \{a_{x,i} \mid x \in X, 1\!\leq\!i\!\leq\!rk(x)\}$. We define two homomorphisms $\phi_1, \phi_2 \colon B(X) \to P(A)$, setting $\phi_1(x^i) = a_{x,i},\: \phi_1(\overline{x}^i) = \overline{a}_{x,i}, \linebreak \phi_2(x^1) = a_{x,1}, \: \phi_2(\overline{x}^{i-1}) = a_{x,i}, \phi_2(x^i) = \overline{a}_{x,i}, \: i \in \overline{2,\aar(x)}, \:\phi_2(\overline{x}^{\aar(x)}) = \overline{a}_{x,0}$. We introduce the factor-monoid $S_X = B(X) / Ker \phi$ where $\phi(x) = \inang{ \phi_1(x), \phi_2(x) } \colon \linebreak B(X) \to \mathcal{P}(A) \times \mathcal{P}(A)$. Then $B(X)$ can be considered as the generating set for $S_X$ and we want to prove that the identity language of $S_X$ is exactly $D(X)$.

Let $w$ be a word representing identity in $S_X$ and $w_1, w_2$ be the words representing its images under $\phi_1, \phi_2$ respectively. Then $w_1$ and $w_2$ represent identity in $\mathcal{P}(A)$. Let $R_1, R_2$ be the binary relations over $Pos(w)$ defined as follows: $(i, j) \in R_l$ iff $w_l[i]$ and $w_l[j]$ contract with each other when reducing the word $w_l$ to identity. Since there is only one ``contracting relation'' for any correct bracket sequence, the relations $R_1, R_2$ are uniquely defined by the word $w$ which represents identity.

\begin{proposition}\label{prop-chain}
Let $x \in X, r = \aar(x)$ and $i_1 < j_1 < \ldots < i_r < j_r$ be such that $w[i_1] = x^0$ and it holds that $(i_l, j_l) \in R_1$ for any $l \leq r$ and $(j_l, i_{l+1}) \in R_2$ for any $l < r$. Then $(i_1, j_r) \in R_2$ and for any $l < r$ it holds that $w[i_l] = x^l, \: w[j_l] = \overline{x}^l$.
\end{proposition}
\begin{proof}
The second statement is established according to the definitions of $\phi_1, \phi_2$ and $R_1, R_2$. It remains to prove the first one. There is a cycle of numbers $p_1 = i_1, q_1 = j_1, \ldots, p_r = i_r, q_r = j_r, p_{r+1}, q_{r+1}, \ldots, p_{2r}, q_{2r}, \ldots, p_{dr}, q_{dr}, p_{dr+1} = p_1$ such that $(p_l, q_l) \in R_1$ and $(q_l, p_{l+1}) \in R_2$ for any $l \leq dr$. We prove that actually $r = 1$ which implies the theorem. Suppose the converse and let $i_1$ be the leftmost element $i$ in this cycle such that $w[i] = x^0$, then $p_{r+1} > p_1$. It is easy to prove by induction on $t$ using the planarity of $R_1, R_2$ that for every $t > r$ there exists some $l \leq r$ such that $i_l < p_t < q_t < j_l$ and $p_{r+1} \leq p_t < q_t < q_r$. This contradicts the equality $p_{dr+1} = p_1$. The proposition is proved.
\end{proof}

Let us call refer as chain cycles the sets consisting of $i_1, j_1, \ldots, i_r, j_r$ from the proposition. They form a partition of $Pos(w)$ since $R_1, R_2$ are total one-to-one relations. The proposition above and the planarity of relations $R_1, R_2$ imply that chain cycles can serve as sets $H_l$ from the definition of multibracket sequence. So we have proved:
\begin{lemma}
Any element of the identity language of $S_X$ with respect to the set $B(X)$ is a correct multibracket sequence over the set $X$.
\end{lemma}

\begin{lemma}
Any correct multibracket sequence over the set $X$ is an element of the identity language of $S_X$ with respect to the set $B(X)$.
\end{lemma}
\begin{proof}
Recall the grammar $G_X$ from the previous section generating the set $D(X)$. To prevent confusion we denote the separator in the grammar by $\#$ instead of $1$ We extend the mappings $\phi_l$ to the set $(B(X) \cup \{\#\})^*$ defining $\phi_1(\#) = \phi_2(\#) = 1$. We denote by $\mu_l(w)$ the value of the word $\phi_l(w)$ in $\mathcal{P}(A)$, obviously $\mu_l$ is a homomorphism. We want to prove by induction that if $S_i \vdash w, \: w = w_0 \# w_1 \ldots \# w_i$, then $\mu_1(w_0) = \mu_1(w_1) = \ldots = \mu_1(w_i) = \mu_2(w) = 1$.

Consider the rule applied in the root of the derivation tree. The basis of induction if the obvious case of the rules $S_0 \to \epsilon$ or $S_1 \to \#$. In case of the rules $S_{i+j} \to S_i \cdot S_j$ and $S_{i+j-1} \to S_i +_k S_{j}$ the induction statement follows from the fact that the inverse homomorphic image of $1$ is closed under concatenation and intercalation.

In the case of the rule $S_i \to x^1(\ldots( S_i  +_1 (\overline{x}^1 \# x^2)+_2\ldots) +_i (\overline{x}^i \# x^{i+1}) \overline{x}^{i+1}$ we consider the components of the word $w$. There exists a word $u = u_0 \# \ldots \# u_{i}$, derivable from $S_i$, such that for any $j \leq i$ it holds that $w_i = x^{i+1} u_i \overline{x}^{i+1}$. So $\mu_1(w_l) = \mu_1(x^{l+1}) \mu_1(u_l) \mu_1(\overline{x}^{l+1}) = a_{x,l+1} 1 \overline{a}_{x,l+1} = 1$. Let us prove $\mu_2(v) = 1$, indeed $\mu_2(v) = \mu_2(x^1) \mu_2(u^0) \mu_2(\overline{x}^1 \# x^2) \mu_2(u^1) \ldots \mu_2(\overline{x}^i \# x^{i+1}) \mu_2(u_i) \mu_2(\overline{x}^{i+1}) = a_{x, 1} \mu_2(u_0)a_{x,2} \overline{a}_{x,2}\mu_2(u_1) \ldots a_{x,i+1} \overline{a}_{x,i+1} \mu_2(u_i) \overline{a}_{x,1} = a_{x,1} \mu_2(u) \overline{a}_{x,1} = a_{x,1} \overline{a}_{x,1} = 1$. The last case is verified and the lemma is proved.
\end{proof}

\begin{theorem}
The class of languages recognized by $k$-DCFGs is exactly the class of languages recognized by $S_X$-automata for the generating sets $X$ of rank $k+1$.
\end{theorem}
\begin{proof}
By the lemmas above the language $S_X$ coincides with the set of multibracket sequences $D(X)$, which is generated by a $(rk(X)-1)$-displacement grammar. Other languages recognized by $S_X$-automata are its images under rational transductions and, hence, displacement context-free languages since the latter are closed under rational transductions. From Theorem \ref{Chomsky} it follows that all $k$-displacement context-free languages are rational transductions of $D(X)$ for some set $X$ of rank $k+1$ and then by Theorem \ref{Kambites} they are all recognized by $S_X$-automata.
\end{proof}

\section{Simultaneous two-stack automata}

In the case of usual bracket sequences the opening and closing brackets naturally correspond to push and pop operations. In the case of multibracket sequences each bracket is in fact a pair of brackets, so every multibracket is an operation on the pair of stacks. The full power of two-stack machines allows to simulate every recursively enumerable language, but in our case there are some restrictions on possible operations. The most principal limitation is that our operations are synchronized: every move changes the length of each stack by $1$. In general, there are only four possible types of operations: push the same symbol on both stacks, move the symbol from the first stack to the second, return a symbol back to the first stack from the second and remove the same symbol from both the stacks. Also the rank of a symbol determines the number of time it is exchanged between the stacks.

Note that Proposition \ref{prop-chain} in fact postulates that if a symbol $a$ of arity $k$ is pushed on the stack together with its copy $a'$ then after transferring it $2(k-1)$ times between the stacks it would be removed together with exactly the same symbol $a'$. Therefore we should trace only the number of exchanges the symbol participated in, so we will keep in stacks not the symbols alone but the pairs consisting of a symbol and a counter of its number of exchanges. This counter is incremented every time the symbol is moved from one stack to another and equals $1$ after the first push. When we proceed the remove operation, we verify that the top element of the first stack is $\inang{a, 2k-1}$ and the top element of the second stack is $\inang{a, 1}$ with the same $a$. We call this model of computation a simultaneous two-stack automata.

\newcommand{\ppu}{\textrm{PUSH} }
\newcommand{\mmo}{\textrm{MOVE} }
\newcommand{\rre}{\textrm{RETURN} }
\newcommand{\ppo}{\textrm{POP} }
\newcommand{\kke}{\textrm{KEEP} }

\begin{definition}
A simultaneous two-stack automaton of rank $k$ ($k$-STSA) is a tuple $\Aa = \inang{Q, \Sigma, \Gamma, \aar, P, q_0, F}$ where $Q$ is a finite set of states, $\Sigma$ is a finite alphabet, $\Gamma$ is a finite stack alphabet, $\aar \colon \Gamma \to \overline{1,k}$ is the arity function, $P$ is the set of transitions, $q_0 \in Q$ is an initial state and $F \subseteq Q$ is a set of final states. Transitions has the form $(\inang{q_1, a} \to \inang{q_2, \tau, \alpha})$, where $q_1, q_2$ are states, $a \in \Sigma \cup \{\epsilon\}$ is an input symbol (or an empty word) and $\tau \in \inang{\ppu, \mmo, \rre, \ppo}$ is a command and $\alpha \in \Gamma$ is a stack symbol.
\end{definition}

As in the case of usual finite automata the formal definition of the acceptance relation is given through the notion of configuration, which is the instantaneous description of the automaton.

\newcommand{\Sigman}{\Sigma_{\bbbn}}
\begin{definition}
A configuration of a simultaneous two-stack automaton $\Aa = \inang{Q, \Sigma, \Gamma, \aar, P, q_0, F}$ is a tuple $\inang{q, u, \beta_1, \beta_2}$ where $q \in Q$ is the current state, $u$ is the current suffix of input, which has not been processed yet and $\beta_1, \beta_2$ are the words in the alphabet $\Sigma_{\bbbn} = \Sigma \times \bbbn$. A transition relation $\vdash_{\Aa}$ is the smallest transitive reflexive relation such that:
\begin{itemize}
  \item If $(\inang{q_1, a} \to \inang{q_2, \ppu, \alpha}) \in P$ then $\inang{q_1, au, \beta_1, \beta_2} \vdash \inang{q_2, u, \beta_1(\alpha,1), \linebreak\beta_2(\alpha,1)}$ for any words $u \in \Sigma^*$ and $\beta_1, \beta_2 \in \Sigman$.
  \item If $(\inang{q_1, a} \to \inang{q_2, \mmo, \alpha}) \in P$ then $\inang{q_1, au, \beta_1 (\alpha,2i-1), \beta_2} \vdash \inang{q_2, u, \beta_1, \linebreak\beta_2(\alpha,2i)}$ for any words $u \in \Sigma^*$ and $\beta_1, \beta_2 \in \Sigman$ and any counter value $i < \aar(\alpha)$.
  \item If $(\inang{q_1, a} \to \inang{q_2, \rre, \alpha}) \in P$ then $\inang{q_1, au, \beta_1, \beta_2(\alpha,2i)} \vdash \inang{q_2, u, \linebreak\beta_1(\alpha,2i+1), \beta_2}$ for any words $u \in \Sigma^*$ and $\beta_1, \beta_2 \in \Sigman$ and any counter value $i < \aar(\alpha)$.
  \item If $(\inang{q_1, a} \to \inang{q_2, \ppo, \alpha}) \in P$ then $\inang{q_1, au, \beta_1(\alpha,2\aar(\alpha)-1), \beta_2(\alpha,1)} \vdash \inang{q_2, u, \beta_1, \beta_2}$ for any words $u \in \Sigma^*$ and $\beta_1, \beta_2 \in \Sigman$.
\end{itemize}
The language $L(\Aa)$ recognized by the automaton equals $L(\Aa) = \inbr{w \in \Sigma^* \mid \exists q \in F (\inang{q_0, w, \epsilon, \epsilon} \vdash \inang{q, \epsilon, \epsilon, \epsilon})}$.
\end{definition}

The condition on counter parity reflects the fact that only the symbols that were moved from the first stack to the second can be returned back. If a symbol was initially pushed to the second stack then it would be removed from it only by the pop operation. That is done in order to keep the structure of multibracket chain which has one \kav{embracing} link in the lower half plane and $k$ small links in the upper half plane. It shows that the stacks are not completely symmetric in their roles, in fact the first stack is basic and the second is just an additional memory register which also stores the placeholders of symbols pushed to the first stack to process the \mmo and \rre operations in a correct order.

\newcommand{\rrule}[4]{\inang{#1, #3} \to \inang{#2, #4}}
\renewcommand{\rule}[5]{\inang{#1, #3} \to \inang{#2, #4, #5}}
\newcommand{\arrayrrule}[4]{\inang{#1, #3} & \to & \inang{#2, #4}}
\newcommand{\arrayrule}[5]{\inang{#1, #3} & \to & \inang{#2, #4, #5}}
Note that the set of possible memory operations can be extended by the $\kke$ command which does not change the contents of the stacks. In order to simulate an edge with $\kke$ between states $q_1$ and $q_2$ we add a dummy state $q'$ in the middle and a new symbol $Z$ of arity $1$ to the stack alphabet and replace the edge under consideration with two transitions $\inang{q_1, a} \to \inang{q', \ppu, Z}$ and $\inang{q', \epsilon} \to \inang{q_2, \ppo, Z}$. Such procedure decreases the number of ``keeping'' edges so we proceed by induction. In the further the assume that rules of the form $\inang{q_1, a} \to \inang{q_2, \kke}$ are also allowed in the set of transitions.

\begin{example}\label{example-copy}
The rank $2$ two-stack simultaneous automaton $\Aa = \inang{\inbr{q_i \mid 0 \leq i \leq 6}, \inbr{a,b}, \inbr{A,B}, \aar, P, q_0, \inbr{q_6}}$ where $\aar(A) = \aar(B) = 2$ with the set of transitions specified below recognizes the crossing copy language $\{ a^m b^n a^m b^n \mid m, n \in \bbbn\}$.
$$
\begin{array}{rclp{4cm}rcl}
    \arrayrule{q_0}{q_0}{a}{\ppu}{A} &\quad& \arrayrrule{q_0}{q_1}{\epsilon}{\kke}\\
    \arrayrule{q_1}{q_1}{b}{\ppu}{B} &\quad& \arrayrrule{q_1}{q_2}{\epsilon}{\kke}\\
    \arrayrule{q_2}{q_2}{\epsilon}{\mmo}{B} &\quad& \arrayrrule{q_2}{q_3}{\epsilon}{\kke}\\
    \arrayrule{q_3}{q_3}{a}{\mmo}{A} &\quad& \arrayrrule{q_3}{q_4}{\epsilon}{\kke}\\
    \arrayrule{q_4}{q_4}{\epsilon}{\rre}{A} &\quad& \arrayrrule{q_4}{q_5}{\epsilon}{\kke}\\
    \arrayrule{q_5}{q_5}{b}{\rre}{B} &\quad& \arrayrrule{q_5}{q_6}{\epsilon}{\kke}\\
    \arrayrule{q_6}{q_6}{\epsilon}{\ppo}{B} &\quad& \arrayrule{q_6}{q_6}{\epsilon}{\ppo}{A}
\end{array}
$$
\end{example}

For the sake of clarity we describe the computation process of this automaton in details. It is not difficult to see that to traverse the edges successfully the input word should be of the form $a^{m_1} b^{n_1} a^{m_2} b^{n_2}$ on the path from $q_0$ to $q_6$, otherwise some of the reading operations would be impossible. Assume we have a word $a^{m_1} b^{n_1} a^{m_2} b^{n_2}$ that is accepted by the automaton, let us prove that $m_1 = m_2$ and $n_1 = n_2$. In the first part of its computation the automaton reads all the $a$'s from the first segment of the word and both of its stacks contain the words $(A,1)^{m_1}$. Afterwards the automaton passes the edge to $q_1$ and reads all the $b$'s from the second segment, so both the stacks contain $(A,1)^{m_1}(B,1)^{n_1}$ when the automaton is entering the state $q_2$. Note that $A$ should be on the top of the first stack in the state $q_3$, so we should move all the $B$'s to the second stack in $q_2$ and the number of necessary moves is exactly $n_1$. Hence the first stack contains $(A,1)^{m_1}$ and the second stack contains $(A,1)^{m_1}(B,1)^{n_1}(B,2)^{n_1}$ before reading the second segment of $a$'s in $q_3$. In $q_3$ the automaton should read all the remaining $a$'s, so the stack contents are $(A,1)^{m_1 - m_2}$ and $(A,1)^{m_1}(B,1)^{n_1}(B,2)^{n_1}(A,2)^{m_2}$ when the automaton is leaving the state $q_3$. In $q_4$ all the $A$'s moved on the previous step should be returned, so the stacks contain $(A,1)^{m_1 - m_2} (A, 3)^{m_2}$ and $(A,1)^{m_1}(B,1)^{n_1}(B,2)^{n_1}$ when the automaton enters $q_5$. Note that in $q_5$ the automaton must read all the $b$'s in the word in order to finish reading. So if this stage is successful the stacks contain $(A,1)^{m_1 - m_2} (A, 3)^{m_2}(B,3)^{n_2}$ and $(A,1)^{m_1}(B,1)^{m_2}(B,2)^{n_1-n_2}$. Since in $q_6$ the automaton executes only $\ppo$ operations there should be no $(A,1)$'s on the first stack and no $(B,2)$'s on the second stack implying that $m_1 = m_2$ and $n_1 = n_2$ which was required. The correctness of the automaton is proved.

Recall the definition of $S_X$-automata from Section \ref{section-monoid-displ}. Since the notion of simultaneous two-stack automata is just a reformulation of $S_X$-automata and the rank of the automata equals the rank of the generating set, the following theorem holds:

\begin{theorem}\label{STSA-theorem}
Simultaneous two-stack automata of rank $k$ recognize exactly the family of $(k-1)$-displacement context-free languages, which is the family of $k$-well-nested multiple context-free languages.
\end{theorem}

It follows that simultaneous two-stack automata of rank $2$ recognize exactly the family of tree-adjoining languages.

\section{Generalized simultaneous two-stack automata}

Though the introduced notion of simultaneous two-stack automata of rank $k$ directly corresponds to the notion of $(k-1)$-displacement context-free language, the formulation itself seems to be not satisfactory. The most disadvantage of the formulation is the lack of flexibility: note that, for example, the recognizing power of context-free languages remains the same, no matter either the lookup of the arbitrary finite number of top stack symbols is allowed, the lookup of only the top symbol is possible or there is no lookup at all. We want to gain the analogous flexibility in our case.

The first inconvenient restriction is that we are bound to push and pop the same symbols from both the stacks and it is not possible, for example, to push $A$ to the first stack and $B$ to the second. Analogously we cannot remove $A$ from the first stack adding $B$ to the second, the pushed symbol must be also $A$. If we weaken this restriction and allow to combine arbitrary symbols in such operations it is impossible to trace the rank of particular element of stack alphabet. However, we still want to distinguish, say, $2$-DCFLs from $3$-DCFLs so the notion of rank cannot be completely omitted. So we keep on associating a counter with every symbol on the stacks and incrementing this counter during every \mmo and \rre operation. This counter is required to be less than $2K$ during the computation, where $K$ is the rank of the automaton. The discussion above leads us to the following definition:

\begin{definition}
A generalized simultaneous two-stack automaton of rank $k$ ($k$-GSTSA) is a tuple $\Aa = \inang{Q, \Sigma, \Gamma, P, q_0, F}$ where $Q$ is a finite set of states, $\Sigma$ is a finite alphabet, $\Gamma$ is a finite stack alphabet, $P$ is the set of transitions, $q_0 \in Q$ is an initial state and $F \subseteq Q$ is a set of final states. Transitions has the form $(\inang{q_1, a} \to \inang{q_2, \tau, \alpha_1, \alpha_2})$, where $q_1, q_2$ are states, $a \in \Sigma \cup \{\epsilon\}$ is an input symbol (or an empty word) $\tau \in \inang{\ppu, \mmo, \rre, \ppo}$ is a command and $\alpha_1, \alpha_2 \in \Gamma$ are stack symbols.
\end{definition}

The notion of configuration for $k$-GSTSAs is the same that for usual $k$-STSAs, the configuration includes the current state, the suffix of input to be read and the contents of the stacks. Since we have changed the format of automaton commands we should also modify the transition relation.

\begin{definition}
A transition relation $\vdash_{\Aa}$ is the smallest transitive reflexive relation such that
\begin{itemize}
  \item If $(\inang{q_1, a} \to \inang{q_2, \ppu, \alpha_1, \alpha_2}) \in P$ then $\inang{q_1, au, \beta_1, \beta_2} \vdash \inang{q_2, u, \beta_1(\alpha_1,1), \linebreak\beta_2(\alpha_2,1)}$ for any words $u \in \Sigma^*$ and $\beta_1, \beta_2 \in \Sigman$.
  \item If $(\inang{q_1, a} \to \inang{q_2, \mmo, \alpha_1, \alpha_2}) \in P$ then $\inang{q_1, au, \beta_1 (\alpha_1,2i-1), \beta_2} \vdash \inang{q_2, u, \beta_1, \beta_2(\alpha_2,2i)}$ for any words $u \in \Sigma^*$ and $\beta_1, \beta_2 \in \Sigman$ and any counter value $i < k$.
  \item If $(\inang{q_1, a} \to \inang{q_2, \rre, \alpha_1, \alpha_2}) \in P$ then $\inang{q_1, au, \beta_1 , \beta_2(\alpha_1,2i)} \vdash \inang{q_2, u, \linebreak\beta_1(\alpha_2,2i+1), \beta_2}$ for any words $u \in \Sigma^*$ and $\beta_1, \beta_2 \in \Sigman$ and any counter value $i < k$.
  \item If $(\inang{q_1, a} \to \inang{q_2, \ppo, \alpha_1, \alpha_2}) \in P$ then $\inang{q_1, au, \beta_1(\alpha_1,2i-1), \beta_2(\alpha_2,1)} \vdash \inang{q_2, u, \beta_1, \beta_2}$ for any words $u \in \Sigma^*, \: \beta_1, \beta_2 \in \Sigman$ and any counter value $i < k$.
\end{itemize}
The language $L(\Aa)$ recognized by the automaton equals $L(\Aa) = \inbr{w \in \Sigma^* \mid \exists q \in F (\inang{q_0, w, \epsilon, \epsilon} \vdash \inang{q, \epsilon, \epsilon, \epsilon})}$.
\end{definition}
Note that we can simulate keeping transitions in the automaton as well as earlier.

We use the values of counters not only to trace the number of \mmo and \rre operations performed in a chain, but also use their parity for the same purpose as in the case of STSA-s. In fact, we want to keep the multibracket geometric structure of the stack contents since this structure reflects the order and embedding of constituents.

Now we want to prove that $k$-GSTSAs have the same recognizing power as $k$-STSAs for any natural $k$. First note that the latter are just a particular case of the former since we can set $\alpha_1 = \alpha_2$ in all the transitions of the automaton. To prove the opposite inclusion we again refer to multibracket sequences. In this case we will not embed this approach into monoid framework to escape unnecessary technicalities.

Let $A = \inbr{a_1, \overline{a}_1, \ldots, a_m, \overline{a}_m}$ be the alphabet of brackets and $Y \subseteq A \times A$ be the set of admissible pairs. For any letter $a \in Y$ we denote by $\pi_i(a), \: i = 1, 2$, its $i$-th coordinate. The mapping $\pi_i$ is naturally extended to words in $Y^*$, we call $\pi_i(w)$ the $i$-th projection of the word $w$. The notion of $k$-garland introduced below is a generalization of the notion of multibracket sequence for the case of arbitrary set $Y$. Recall that if $u$ is a correct multibracket sequence, then the contraction relation $R(u)$ consists of all such pairs $\inang{i,j}$ that the letters $u[i]$ and $u[j]$ contract with each other in $u$ when reducing it to an empty word. Note that $R$ is always a symmetric bijection and for every correct bracket sequence there is only one such relation. We define also an asymmetric contraction relation $R_<(u)$; a pair $\inang{i,j}$ belongs to $R_<(u)$ if it belongs to $R(u)$ and the inequality $i < j$ holds.

\begin{definition}\label{garland-def}
The word $w \in Y^*$ is a $k$-garland over the alphabet $Y$ if the following conditions hold:
\begin{enumerate}
\item $\pi_1(w), \pi_2(w)$ are correct bracket sequences.
\item If $i_1, j_1, i_2, j_2$ are indexes such that $j_1 < i_2$, $(i_1, j_1), (i_2, j_2) \in R(\pi_1(u))$ and $(j_1, i_2) \in R(\pi_2(u))$ then either $i_1 < j_1 < i_2 < j_2$, $j_1 < i_1  < i_2 < j_2$ or $i_1 = j_2$ (in this case also $j_1 = i_2$).
\item If $i_1 < j_1 < i_2 < j_2 < \ldots < i_l < j_l$ is an ascending chain of indexes such that $(i_t, j_t) \in R(\pi_1(w))$ for any $t \leq l$ and $(j_t, i_{t+1})$ for any $t < l$ then the inequality $l \leq k$ holds.
\end{enumerate}
\end{definition}

Let $R_0(w)$ define the relation $(R_<(\pi_1(w)) \cup R(\pi_2(w)))^*$. Then the following lemma holds:

\begin{lemma}\label{garland-cycle-ex}
Any vertex in the set $Pos(w) = \overline{0,|w|-1}$ belongs to some simple cycle in the graph $G_R = \inang{Pos(w), R_0}$.
\end{lemma}
\begin{proof}
Since the number of vertexes is finite, it suffices to proof that every edge in $R_0$ belongs to some infinite path with no edges traversed in both directions. Then it suffices to show that there is in infinite path in $G_R$ with the edges from $R_<(\pi_1(w))$ (we call them the edges of the second type) and the edges from $R(\pi_2(w))$ (the edges of the second type) being alternated. Let us start from an arbitrary edge $(i_1, j_1)$ of the first type and show we can always add two more edges. Indeed, there is some edge $(j_1,i_2)$ of the second type because the $R(\pi_2(w))$ is a bijection. Then there is an edge $(i_2, j_2) \in R(\pi_1(w))$, we need to show that $i_2 < j_2$. In both the cases it follows from the second part of the definition of $k$-garland. Then we have added to more edges to the path and the lemma is proved.
\end{proof}

\begin{lemma}\label{garland-cycle}
If $w$ is a $k$-garland, then every vertex $i \in Pos(w)$ belongs to some cycle in the graph $G_R = \inang{Pos(w), R}$ containing the indexes $i_1 < j_1 < \ldots < i_l < j_l$ such that for any $t \leq l$ it holds that $(i_t, j_t) \in R(\pi_1(w))$ and for any $t < l$ it holds that $(j_t, i_{t+1})$ belongs to $R(\pi_2(w))$. It also holds that $(j_t, i_1) \in R(\pi_2(w))$ and $l \leq k$.
\end{lemma}
\begin{proof}
Consider the cycle which contains $i$, such a cycle exists due to Lemma \ref{garland-cycle-ex}. Take the leftmost vertex $i_0$ in this cycle and consider the longest ascending path containing $i_0$, according to the definition of $R_0(w)$ it starts and ends with en edge of the first type. Then the proof of the statement $(j_t, i_1) \in R(\pi_2(w))$ repeats the proof of the Proposition \ref{prop-chain}. The condition $l \leq k$ follows from the definition of $k$-garland.
\end{proof}

Since the structure of states is the same for automata of all kinds, we should concentrate on the structure of their transitions. Let $\mathcal{T}$ be some transition of the generalized two-stack simultaneous automaton $\Aa = \inang{Q, \Sigma, \Gamma, P, q_0, F}$. Its stack image of $\psi(\mathcal{T})$ is a pair of symbols in the alphabet $\Gamma \cup \{\overline{A} \mid A \in \Gamma\}$ defined as follows:
\begin{enumerate}
  \item If $\mathcal{T} = (\inang{q_1, a} \to \inang{q_2, \ppu, \alpha_1, \alpha_2})$ then $\psi(\mathcal{T}) = \inang{\alpha_1, \alpha_2}$,
  \item If $\mathcal{T} = (\inang{q_1, a} \to \inang{q_2, \mmo, \alpha_1, \alpha_2})$ then $\psi(\mathcal{T}) = \inang{\overline{\alpha}_1, \alpha_2}$,
  \item If $\mathcal{T} = (\inang{q_1, a} \to \inang{q_2, \rre, \alpha_1, \alpha_2})$ then $\psi(\mathcal{T}) = \inang{\alpha_1, \overline{\alpha}_2}$,
  \item If $\mathcal{T} = (\inang{q_1, a} \to \inang{q_2, \ppo, \alpha_1, \alpha_2})$ then $\psi(\mathcal{T}) = \inang{\overline{\alpha}_1, \overline{\alpha}_2}$.
\end{enumerate}

We denote by $\psi(\Aa) = \{ \psi(\mathcal{T} \mid \mathcal{T} \in P\}$ the set of stack images of the transitions of the automaton $\Aa = \inang{Q, \Sigma, \Gamma, P, q_0, F}$. Two transitions of the GSTSA are called consecutive if the destination set of the first transition equals the source set of the second one. We call a computation a sequence of consecutive transitions. The computation is identity-preserving if there is nothing in the stacks after its termination provided the stacks are empty before it starts. Note that a word $w$ is accepted by an automaton iff there is an identity-preserving computation of this automaton which starts in the initial state, terminates in some of the final states and reads exactly the word $w$.

\begin{definition}
The stack image $\psi(\mathcal{C}))$ of the computation $\mathcal{C} = \mathcal{T}_1 \ldots \mathcal{T}_r$ is the sequence $\psi(\mathcal{T}_1) \ldots \psi(\mathcal{T}_r)$.
\end{definition}

\begin{proposition}\label{id-garlands}
The identity-preserving computations of the $k$-GSTSA $\Aa = \inang{Q, \Sigma, \Gamma, P, q_0, F}$ are exactly all $k$-garlands over the set $\psi(\Aa)$.
\end{proposition}
\begin{proof}
Consider some sequence of \kav{push} and \kav{pop} operations executed on a single stack. The emptiness of the stack if preserved under this sequence of operations iff the sequence maps to a correct bracket sequence under a natural encoding of operations. Since the projections of $k$-garlands are correct bracket sequences every $k$-garland is identity-preserving.

The opposite implication uses the specificity of $k$-GSTSA operations. Let a computation be identity-preserving then the first part of the $k$-garland definition is obviously valid. Let $R_i, \: i= 1, 2$ denote the contraction relation of the sequence of operations on the $i$-th stack. If $(i_1, j_1), (i_2, j_2) \in R_1, \: (j_1, i_2) \in R_2$ and $j_1 < i_2$; it means that in the $i_2$-th step of the computation we pop from the second stack the element pushed there on the $j_1$-th step. There are two possibilities: first, if this pop is a part of the \rre operation then by the definition of GSTSA only the \mmo operation is possible in the $j_1$-th transition of the computation, also the symbol pushed on the first stack during the \rre operation must be removed somewhen later. It means that $i_1 < j_1$ and $i_2 < j_2$. The second variant is that the \ppo operation is executed on the $i_2$-th step, it implies that the operation on the step $i_1$ is \ppu which implies $j_1 < i_1$ and $j_2 < i_2$. Both possibilities are allowed in the definition of $k$-garland so the second step is proved. To prove the third part of the definition note that all the intermediate elements of the ascending chains considered in that part are linked by \mmo and \rre operations. Since every such operation increments the value of the same counter the number of intermediate operations is not greater then $2k - 2$ and the total number of vertexes in this chain is not greater then $2k$ which was required. The lemma is proved.
\end{proof}

\begin{corollary}\label{GSTSA-trans}
For any $k$-GSTSA $\Aa$ the language $L(\Aa)$ is a rational transduction of the set of $k$-garlands over the alphabet $\psi(\Aa)$.
\end{corollary}
\begin{proof}
Evidently $L(\Aa)$ is the rational transduction of the set of identity-preser\-ving computations. Then we should apply the Proposition \ref{id-garlands}.
\end{proof}

\begin{lemma}\label{STSA-garland}
The set of $k$-garlands over the alphabet $\psi(\Aa)$ is recognized by some $k$-STSA.
\end{lemma}
\begin{proof}
Consider the finite set $\mathcal{D}$ of all possible closed chains in $k$-garlands and some chain $d \in \mathcal{D}$. Let $l(d)$ denote its number of vertexes in the chain and $d[i]$ denote its $i$-th leftmost vertex. Consider $\mathcal{D}$ as the ranked alphabet with the arity function $l$ and define the set of multibrackets $B(X) = \inbr{d[i] \mid d \in \mathcal{D}, 1 \leq i \leq l(D)}$. It is easy to prove that
the set of $k$-garlands is the homomorphic image of the generalized Dyck language $D(\mathcal{D})$ of correct multibracket sequences which is a $(k-1)$-DCFL. Then it is recognized by some $k$-STSA due to Theorem \ref{STSA-theorem}.
\end{proof}

\begin{theorem}\label{GSTSA-STSA}
Any language recognized by some $k$-GSTSA is recognized by some $k$-STSA.
\end{theorem}
\begin{proof}
The languages recognized by $k$-STSAs are closed under rational transductions. By Corollary \ref{GSTSA-trans} it suffices to show that the language of $k$-garlands is recognized by a $k$-STSA which was proved in Lemma \ref{STSA-garland}.
\end{proof}

We have proved that the permission for STSA commands to combine arbitrary pairs of symbols does not affect its recognizing power. It is worth mentioning that in fact $k$-garlands are a natural generalization of multibracket sequence under the same permission so the method of the Section \ref{section-monoid-displ} can also be used to find another version of Chomsky-Sch\"{u}tzenberger theorem for the DCFLs.

\section{Blind and sighted automata}

There is another major disadvantage in our initial definition of STSA: the automaton is not able to observe top symbols on the stacks during its computation. Certainly, these symbols are significant in the case of \ppo operation since the automaton halts if the command to execute is, say, $\inang{\ppo, A, B}$ and current top symbols are $C$ and $D$. In the same way the \mmo command takes into account the content of the first stack, as well as the \rre operation --- of the second. However, there is no possibility to refer to the top elements of the stack in the case of \ppu operation. This limitation seems to be unnatural and unpleasant, so we should find some way of modifying the automaton to overcome this difficulty.

Let us first discuss the same problem in the case of usual pushdown automaton. Assume we have a command of the kind \kav{in the state $q_1$ if $A$ is the top symbol of the stack then read $a$ from the input stream, push $B$ to the stack and move to the state $q_2$} (we abbreviate this by $\inang{q, a, A} \to \inang{q_2, \ppu, B}$). The common way to simulate this instruction is to create two fresh states $q'$ and $q''$ and add the following transitions: $(\inang{q_1, \epsilon, \epsilon} \to \inang{q', \ppo, A}), (\inang{q', \epsilon, \epsilon} \to \inang{q'', \ppu, A})$ and $ (\inang{q'', a, \epsilon} \to \inang{q_2, \ppu, B})$. However, it is troublesome to adapt this approach to $k$-GSTSA since it is hard to ensure that the number of move/return operations would not exceed $k$. Therefore we choose another way to simulate top symbol observations.

Let $k$ be the maximal number of stack symbols which are observed in the transitions of the pushdown automaton. Then it has the transition of the following two forms, where $l$ is a natural number not greater than $k$:
$$
\begin{array}{rcl}
  \inang{q_1, a, A_{l+1}\ldots A_k} & \to & \inang{q_2, \ppu, B}\\
  \inang{q_1, a, A_{l+1}\ldots A_k} & \to & \inang{q_2, \ppo, A_k}\\ 
\end{array}
$$
Let $\Gamma$ be the stack of old stack symbols and $Q$ be the set of states. First, we enrich the set of stack symbols with $k$ new symbols $Z_1, \ldots Z_k$ which serve as bottom markers and treat them as elements of $\Gamma$. Then the new set of states is $Q' = \inbr{q_0, q_f} \cup Q \times \Gamma^n$ and the new stack alphabet is $\Gamma \times \Gamma^n$. $q_0$ and $q_f$ are distinguished initial and final states, respectively, and the second component of all other states contains the top $k$ symbols of the stack. Analogously, the second component of the stack symbol always keeps the $k$ symbols below it starting from the deepest. The symbols $Z_1, \ldots Z_k$ were added in order to ensure that there are always at least $k$ symbols in the stack. Then it is straightforward to simulate the dependence from $k$ top symbols by the means of the states only, the only difficulty is to maintain the invariant we announced.

The automaton always starts from the initial state $q_0$ and pushes to the stack the symbols $Z_1 \ldots Z_k$ moving to the state $(q_0, Z_1 \ldots Z_k)$ to start the computation. Every transition of the form $\inang{q_1, a, A_{l+1}\ldots A_k} \to \inang{q_2, \ppu, B}$ is simulated by a new transition 
$$
    \inang{(q_1, A_1 \ldots A_k), a} \to \inang{(q_2, A_2 \ldots A_k B), \ppu, (B, A_1 \ldots A_k)}.
$$
Note that the deepest of the symbols observed on the previous stage in the first component of the automaton state is now observed as the deepest symbol of the second component of the stack top. That allows us to update the top $k$ symbols when the \ppo operation is executed: every transition of the form $\inang{q_1, a, A_{l+1}\ldots A_k} \to \inang{q_2, \ppo, A_k}$ is replaced by the transition
$$
    \inang{(q_1, A_1 \ldots A_k), a} \to \inang{(q_2, A_0A_1 \ldots A_{k-1}), \ppo, (A_k, A_0 \ldots A_{k-1})}.
$$
It is straightforward to prove that the desired invariant is maintained. In the end of the computation we should remove the bottom markers, so we add the transitions of the form $\inang{(q, Z_1 \ldots Z_k), \epsilon} \to \inang{q_f, \ppo, Z_1 \ldots Z_k}$ (it is trivial to simulate immediate pop of $k$ symbols by successively removing them one by one so we simplify the notation) for all former final states $q$. Then it is easy to prove that the new automaton without lookup recognizes exactly the same language as the old automaton did.

Then the same approach can be applied to $k$-GSTSAs. The only modification to be made is to trace the contents of both the stacks, not the single one. So we have proved the following theorem:

\begin{theorem}
The generating power of $k$-GSTSAs is the same whether or not it is allowed to take into account the top $k$ symbols.
\end{theorem} 

\section{Conclusions and future work}

We give the algebraic interpretation of Chomsky-Sch\"{u}tzenberger theorem for the class of displacement-context free languages which are another realization of well-nested multiple context-free languages. We present their characterization in terms of monoid automata. Then we introduce the computational interpretation of the introduced monoid, showing how the multiplication operation of the monoid can be simulated on two stacks by specific combinations of \ppu and \ppo operations. The flexibility of the introduced notion of two-stack automata shows the vitality of the our approach.

There are at least two directions of future work: the first is two develop fast analyzers for the class of DCFGs or for a significant subclass of them. For example, it is interesting to adopt the machinery of LR or Earley algorithms for DCFLs (see \cite{Kanazawa2008} for the variant of Earley analyzer for well-nested MCFGs). The other direction is the further investigation of underlying algebraic structure. The most straightforward question is to provide the same characterization in terms of monoids for the variants of generalized STSAs as it is done for simple STSAs. Also it is interesting to answer, whether the $\epsilon$-moves are redundant, like it was done by Zetzsche for automata based on graph products of polycyclic monoids (\cite{Zetzsche2013}).

\end{document}